%% file: main.tex
\documentclass[runningheads]{llncs}
\usepackage[T1]{fontenc}

\usepackage{graphicx}
\usepackage{amsmath, amssymb}
\usepackage{enumitem}
\usepackage{hyperref}
\usepackage{color}
\usepackage{comment}
\usepackage{wrapfig}

\begin{document}
\title{Physics of Information Geometry - Part I: Principle of Least Action on the Probability Simplex}
%
%
\author{C.~Emre Koksal\inst{1}
\and Deniz Sargun\inst{2}\thanks{This work does not relate to the author' s position at Amazon.}
\authorrunning{Koksal and Sargun}
\institute{The Ohio State University, Columbus OH 43210, USA,
\email{koksal.2@osu.edu} \and
Amazon.com Inc., Palo Alto CA 94301, USA, \email{denizsargun@gmail.com}}}
%
%
%

\maketitle
%
\input{0-abstract}
\input{1-introduction-revised-technical}
\input{2-model-revised}
\input{3-main-insight-pythagorean}
\input{4-problem-statement-revised}
\input{5-approach-least-action-information-projection}
\input{6-numeric-example}
\input{7-geodesic-cost-variational-distribution}
\input{8-discussion-conclusion}


%
%
\bibliographystyle{splncs04}
\bibliography{bibliography}
%
%
\begin{appendix}
\input{9-proof_solution}
\end{appendix}

\end{document}

%% file: 0-abstract.tex
\begin{abstract}

We develop a least-action framework for describing how a probability distribution can evolve from an equilibrium state to a prescribed nonequilibrium state under constrained incremental changes.
Taking a Gibbs distribution as the equilibrium reference, the framework gives a direct physical meaning to the geometry of the probability simplex: distance from equilibrium corresponds to nonequilibrium free energy, while changes between successive distributions carry an informational kinetic cost. The Pythagorean structure of relative entropy then provides the central insight of the work. It shows that intermediate distributions chosen via sequential information projections can reduce the kinetic cost of large transitions and establishes an energy-conservation-like relation between the kinetic expenditure along a path and the free energy accumulated in reaching the target distribution. Motivated by this geometry, we construct a greedy least-action path through successive information projections, obtain a closed-form characterization of each projection through the Lambert \(W\) function, and establish a finite-step performance guarantee. We further show that state-dependent costs can be incorporated naturally by reshaping the underlying Gibbs reference, providing a thermodynamic interpretation of path penalties as modifications of the effective energy landscape. Together, these results provide a unified view of distributional evolution through least action, information geometry, and nonequilibrium thermodynamics.


\end{abstract}

%% file: 1-introduction-revised-technical.tex
\section{Introduction}

A broad class of problems in learning, inference, control, and nonequilibrium statistical mechanics can be viewed as motion on the probability simplex: a system begins at one distribution and reaches another through a sequence of admissible changes. In diffusion-based generative models \cite{Denoising}, for example, distributions evolve between a simple reference law and a structured target; related distributional dynamics arise in Bayesian inference, stochastic control, and active inference~\cite{friston2010free,friston2006free}. Beyond identifying the endpoint, a fundamental question is therefore how the geometry of the simplex determines the path between distributions when changes are constrained locally.

We develop a discrete variational description of such paths by combining the thermodynamic interpretation of relative entropy with information-geometric projection. Taking a Gibbs distribution \(q_0\) as the equilibrium reference, we use \(D_{\mathrm{KL}}(p\|q_0)\) as the dimensionless nonequilibrium free energy\footnote{For a nonequilibrium distribution \(p\), the relative entropy \(D_{\mathrm{KL}}(p\|q_0)\), when multiplied by the thermodynamic factor \(k_B T\), is precisely the excess nonequilibrium free energy relative to Gibbs equilibrium.} and KL divergence between successive distributions as an informational kinetic energy. The orientation of the divergences is consequential: it aligns the free-energy quantity with its statistical-mechanical meaning while simultaneously placing the local motion constraint in the form required by the information-geometric Pythagorean relation. This yields a \textbf{least-action construction} in which the geometry of KL divergence directly relates incremental motion on the simplex to changes in free energy. The same viewpoint also provides a natural way to incorporate additional state-dependent costs: when such a cost enters through its expectation under the variational distribution, it augments the free-energy landscape and can be absorbed into an effective Gibbs reference.

The central geometric argument is obtained by constructing the path backward from the preferred distribution \(p_{\mathrm{pref}}\) toward the Gibbs reference \(q_0\). Each point is obtained by an \textbf{information projection} onto a KL ball determined by the kinetic-energy budget of the corresponding step. The Pythagorean inequality then identifies the free-energy increment generated by the projection and lower bounds it directly by the KL cost of the step. \textit{Iterating these projections produces a discrete least-action path together with an energy-conservation-like relation across the trajectory.} The total free-energy increase from equilibrium to \(p_{\mathrm{pref}}\) is \(D_{\mathrm{KL}}(p_{\mathrm{pref}}\|q_0)\), while the Pythagorean decomposition accounts for this increase through the sequence of kinetic increments and projection residuals. This gives \(D_{\mathrm{KL}}(p_{\mathrm{pref}}\|q_0)\) a dual role: it specifies the free-energy separation between the endpoints and controls the total informational resource required to traverse the simplex. In particular, a fixed per-step kinetic budget yields a direct upper bound on the number of information projections required to connect the two distributions. The resulting relation provides a \textbf{discrete analogue of energy conversion}, with the informational kinetic budget along the path accounting for the free energy accumulated in reaching the nonequilibrium target.

The local projection problem is convex and can be characterized explicitly. The thermodynamically aligned ordering of the KL arguments leads to a nonlinear optimality condition whose closed-form solution is expressed through the Lambert \(W\) function. The Lambert-\(W\) update is therefore a direct consequence of the particular information-geometric and thermodynamic structure imposed by the least-action construction.

The Gibbs-referenced formulation also provides a natural interpretation of costs associated with intermediate distributions. \textit{When an expected state cost is added to the free-energy landscape, its effect can be absorbed into a modified Gibbs reference via an exponential tilt, rather than requiring a different geometric construction.} Physically, this means that additional penalties can be viewed as reshaping the underlying energy landscape: states carrying larger cost become exponentially less favorable under the effective equilibrium distribution. When the cost itself is expressed in physical energy units, the tilted distribution takes the usual Gibbs form associated with an effective Hamiltonian, driven by the cost. This connects application-specific path penalties directly to thermodynamic modifications of the underlying equilibrium model while preserving the same information-projection and least-action structure.

Our construction is related to several established geometric descriptions of distributional dynamics. Classical information geometry~\cite{amari2016information,rao1945information,cencov1982statistical} characterizes statistical manifolds through divergence-induced geometry and dual geodesic structures, while Schr\"odinger bridges and entropy-regularized optimal transport~\cite{leonard2014survey,benamou2000computational,villani2009optimal} formulate interpolation between endpoint distributions through dynamical or path-space optimization. Active inference~\cite{parr2022active,dacosta2020active,devries2025efe} provides another important example in which free-energy functionals organize belief and action dynamics. The present construction differs in its objective and KL orientation: the preferred distribution occupies the first argument of the relative entropy, making the equilibrium-referenced quantity directly coincide with the statistical-mechanical free-energy difference. Accordingly, our emphasis is not on variational inference per se, but on the information geometry and thermodynamics of paths between distributions.

The main contributions of this work are:
\begin{itemize}
   \item \textbf{A least-action principle for information geometry:}
    We introduce a discrete mechanics for motion on the probability
    simplex in which thermodynamic free energy defines the potential
    landscape and KL divergence quantifies the kinetic cost of changing
    distributions. Information projection and the Pythagorean geometry
    of relative entropy then determine a trajectory between equilibrium
    and a prescribed nonequilibrium state. The resulting formulation
    yields an energy-conservation-like relation across the path, an
    explicit bound on the number of admissible steps, and a closed-form
    characterization of each optimal projection through the Lambert
    \(W\) function. The same variational structure accommodates
    state-dependent costs through an exponential deformation of the
    Gibbs reference distribution.

    \item \textbf{A physical interpretation of distributional evolution:}
    The framework connects the geometry of statistical distributions to
    the thermodynamics of creating nonequilibrium structure: the
    free-energy separation \(D_{\mathrm{KL}}(p_{\mathrm{pref}}\|q_0)\)
    simultaneously characterizes the energetic value of the target
    distribution and the informational resource governing access to it.
    State-dependent penalties acquire an equally direct interpretation
    as modifications of the effective energy landscape and hence of the
    associated Gibbs equilibrium. This suggests a broader view of
    learning and inference as constrained motion through distribution
    space, in which both the thermodynamic separation of the endpoints
    and the energetic landscape encountered along the path shape the
    complexity of reaching a desired statistical state.
\end{itemize}

%% file: 2-model-revised.tex
\section{Model}
\label{sec:model}

We consider discrete motion on the probability simplex
\begin{equation}
    \Delta(\mathcal X)
    =
    \left\{
    q\in\mathbb{R}_{\geq 0}^{|\mathcal X|}
    :
    \sum_{x\in\mathcal X}q(x)=1
    \right\},
\end{equation}
where \(\mathcal X\) is a finite state space.  The distribution \(q_0\) denotes a
Gibbs equilibrium distribution and serves as the thermodynamic reference state.
A sequence \(\{q_t\}_{t\geq 0}\) describes a discrete trajectory on the simplex,
with \(p_{\mathrm{pref}}\) denoting the prescribed nonequilibrium distribution to
be reached.  Throughout, we use dimensionless energy quantities; multiplication
by \(k_B \tau \) restores physical energy units, where $k_B$ is the Boltzmann constant and $\tau$ is the system temperature.

We begin by defining the free-energy quantity associated with two points on the
simplex.

\begin{definition}[Relative free energy]
For \(p,q\in\Delta(\mathcal X)\), the relative free energy of \(p\) with respect
to \(q\), relative to the Gibbs reference \(q_0\), is
\begin{equation}
    F_r(p\|q)
    \triangleq
    D_{\mathrm{KL}}(p\|q_0)
    -
    D_{\mathrm{KL}}(q\|q_0).
    \label{eq:relative_free_energy}
\end{equation}
\end{definition}
This definition has a direct thermodynamic interpretation.  With \(q_0\) chosen
as the Gibbs distribution, \(D_{\mathrm{KL}}(q\|q_0)\) is the dimensionless
nonequilibrium free-energy excess of \(q\) above equilibrium.  Hence
\(F_r(p\|q)\) is the difference between the nonequilibrium free energies
associated with \(p\) and \(q\).  In particular,
\begin{align}
    F_r(p\|q_0)
    &=D_{\mathrm{KL}}(p\|q_0),                                      \label{eq:fr_from_equilibrium}\\
    F_r(p\|p)
    &=0,                                                            \label{eq:fr_same_state}\\
    F_r(q_0\|q)
    &=-D_{\mathrm{KL}}(q\|q_0),                                     \label{eq:fr_to_equilibrium}
\end{align}
and the antisymmetry
\begin{equation}
    F_r(p\|q)=-F_r(q\|p)
    \label{eq:fr_antisymmetry}
\end{equation}
follows immediately.  Thus, unlike KL divergence itself, relative free energy
is signed: it records whether the first distribution lies above or below the
second in the free-energy landscape defined by the Gibbs reference.

For the preferred distribution, this becomes
\begin{equation}
    F_r(p_{\mathrm{pref}}\|q)
    =
    D_{\mathrm{KL}}(p_{\mathrm{pref}}\|q_0)
    -
    D_{\mathrm{KL}}(q\|q_0).
    \label{eq:preferred_relative_free_energy}
\end{equation}
The first term is fixed by the endpoints.  Consequently, as \(q\) varies over
the simplex, the relative free energy with respect to \(p_{\mathrm{pref}}\) is
determined entirely by the nonequilibrium free energy stored at \(q\).

For consecutive points on a trajectory, we similarly write
\begin{equation}
    F_r(q_t\|q_{t-1})
    =
    D_{\mathrm{KL}}(q_t\|q_0)
    -
    D_{\mathrm{KL}}(q_{t-1}\|q_0).
    \label{eq:incremental_relative_free_energy}
\end{equation}
Equation~\eqref{eq:incremental_relative_free_energy} is the free-energy increment
generated by the transition \(q_{t-1}\rightarrow q_t\).  Summing these
increments along any discrete path produces the telescoping relation
\begin{equation}
    \sum_{t=1}^{T}F_r(q_t\|q_{t-1})
    =
    D_{\mathrm{KL}}(q_T\|q_0)
    -
    D_{\mathrm{KL}}(q_0\|q_0)
    =
    D_{\mathrm{KL}}(q_T\|q_0).
    \label{eq:free_energy_telescoping}
\end{equation}
Hence, if \(q_T=p_{\mathrm{pref}}\), the accumulated relative free energy along
the trajectory is exactly
\(D_{\mathrm{KL}}(p_{\mathrm{pref}}\|q_0)\), independent of the particular
discretization of the path.  This endpoint quantity will play a central role in
the geometric construction developed in the following sections.

We next associate an informational kinetic energy with motion between
successive distributions.

\begin{definition}[Kinetic energy]
The kinetic energy of the transition from \(q_{t-1}\) to \(q_t\) is
\begin{equation}
    K_t
    \triangleq
    D_{\mathrm{KL}}(q_t\|q_{t-1})
    =
    \sum_{x\in\mathcal X}
    q_t(x)
    \log\frac{q_t(x)}{q_{t-1}(x)}.
    \label{eq:kinetic_energy}
\end{equation}
\end{definition}
The distinction between the two energy quantities is structural.
Relative free energy measures displacement in the thermodynamic landscape
relative to the fixed Gibbs state \(q_0\), whereas \(K_t\) measures the local
cost of the transition itself.  In the infinitesimal limit, KL divergence
induces the Fisher--Rao metric, so that
\begin{equation}
    D_{\mathrm{KL}}(q_t\|q_{t-1})
    =
    \frac{1}{2}
    \|q_t-q_{t-1}\|_{F,q_{t-1}}^2
    +
    o\!\left(\|q_t-q_{t-1}\|^2\right),
    \label{eq:kl_fisher_local}
\end{equation}
providing the usual quadratic dependence on an infinitesimal displacement that
motivates its interpretation as kinetic energy on the statistical manifold.

The two definitions therefore separate the thermodynamics of \emph{where the
system is} from the information-geometric cost of \emph{how it moves}: the
Gibbs-referenced KL divergence assigns free energy to points on the simplex,
while the KL divergence between consecutive points assigns kinetic energy to
steps along a path.  The next section shows that information projection links
these quantities directly through the Pythagorean geometry of relative entropy.

%% file: 3-main-insight-pythagorean.tex
\section{Main Insight: Pythagorean Geometry of Free-Energy Conversion}
\label{sec:pythagorean}

The definitions in Section~\ref{sec:model} assign two distinct energetic
quantities to a trajectory on the probability simplex: the Gibbs-referenced
relative free energy associated with its points and the KL-based kinetic energy
associated with its steps. We now show that these quantities are linked
directly by the Pythagorean geometry of relative entropy. This relation is the
main geometric insight underlying the path construction developed in the
remainder of the paper.

The construction is most naturally viewed backward. Starting from a distribution
\(q_t\) (with the ultimate destination \(q_T=p_{\mathrm{pref}}\)), we construct its predecessor
toward the Gibbs equilibrium \(q_0\). Define
\begin{equation}
\mathcal B_\delta(q_t)
=
\left\{q\in\Delta(\mathcal X):
D_{\mathrm{KL}}(q_t\|q)\leq\delta\right\}.
\label{eq:backward_kl_ball}
\end{equation}
The predecessor is the information projection toward \(q_0\),
\begin{equation}
q_{t-1}^{\star}
=
\underset{q\in\mathcal B_\delta(q_t)}{\text{argmin}} \ ~
D_{\mathrm{KL}}(q\|q_0).
\label{eq:backward_information_projection}
\end{equation}

Figure~\ref{fig:pythagorean-simplex} illustrates the geometry. The KL ball is
centered at \(q_t\) in the divergence sense
\(D_{\mathrm{KL}}(q_t\|q)\leq\delta\), and \(q_{t-1}^{\star}\) is selected by
projection toward \(q_0\). Read forward, the radius
\(D_{\mathrm{KL}}(q_t\|q_{t-1}^{\star})\) is exactly the kinetic energy of the
step.

\begin{figure}[t]
\centering
\includegraphics[width=0.75\columnwidth]{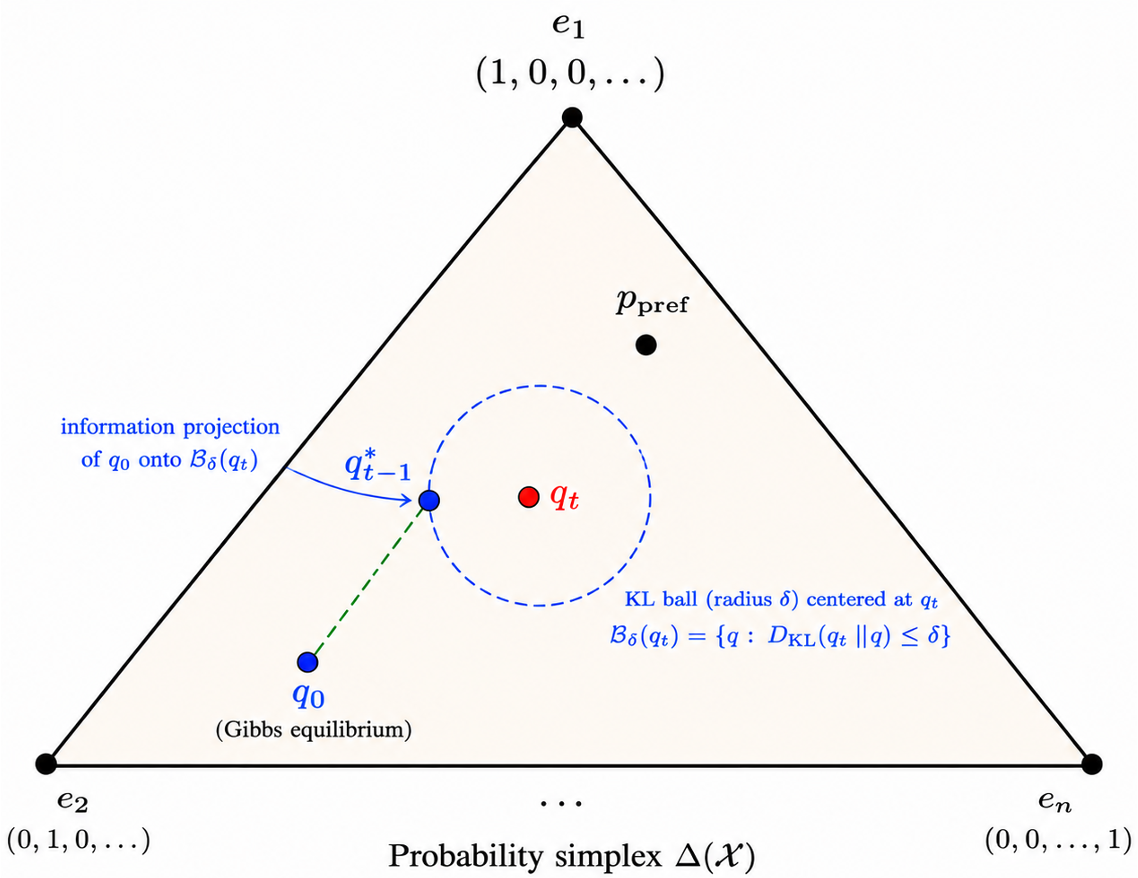}
\caption{Backward information-projection construction. The KL ball
\(\mathcal B_\delta(q_t)=\{q:D_{\mathrm{KL}}(q_t\|q)\leq\delta\}\) is centered
at \(q_t\), and \(q_{t-1}^{\star}\) is its information projection toward the
Gibbs distribution \(q_0\). Read in the forward direction,
\(D_{\mathrm{KL}}(q_t\|q_{t-1}^{\star})\) is the kinetic energy of the step,
while the Pythagorean inequality relates this cost directly to the corresponding
increase in Gibbs-referenced free energy.}
\label{fig:pythagorean-simplex}
\end{figure}

Although the construction proceeds backward, the physical trajectory is read
forward as \(q_0\rightarrow q_1^\star\rightarrow\cdots\rightarrow
q_T=p_{\mathrm{pref}}\). Hence the ball radius is precisely the kinetic-energy
budget of the corresponding forward transition,
\begin{equation}
K_t=D_{\mathrm{KL}}(q_t\|q_{t-1}^{\star})\leq\delta.
\label{eq:kinetic_ball_radius}
\end{equation}

\subsection{Pythagorean Decomposition}

\begin{theorem}[Free-energy gain under information projection]
\label{thm:pythagorean_free_energy}
Let \(q_{t-1}^{\star}\) be the information projection in
\eqref{eq:backward_information_projection}. Then
\begin{equation}
D_{\mathrm{KL}}(q_t\|q_0)
\geq
D_{\mathrm{KL}}(q_t\|q_{t-1}^{\star})
+
D_{\mathrm{KL}}(q_{t-1}^{\star}\|q_0).
\label{eq:pythagorean_main}
\end{equation}
Equivalently,
\begin{equation}
F_r(q_t\|q_{t-1}^{\star})
\geq
D_{\mathrm{KL}}(q_t\|q_{t-1}^{\star})
=
K_t.
\label{eq:free_energy_kinetic_bound}
\end{equation}
Thus, along the forward trajectory, the kinetic energy expended in an optimal
step is upper bounded by the relative free-energy increase produced by that
step.
\end{theorem}

\begin{proof}
The KL Pythagorean inequality associated with the information projection
\eqref{eq:backward_information_projection} gives
\eqref{eq:pythagorean_main}. Subtracting
\(D_{\mathrm{KL}}(q_{t-1}^{\star}\|q_0)\) from both sides gives
\begin{equation}
D_{\mathrm{KL}}(q_t\|q_0)
-
D_{\mathrm{KL}}(q_{t-1}^{\star}\|q_0)
\geq
D_{\mathrm{KL}}(q_t\|q_{t-1}^{\star}).
\end{equation}
The left-hand side is \(F_r(q_t\|q_{t-1}^{\star})\), while the right-hand side
is \(K_t\), proving the result.
\end{proof}

The orientation of the divergences in \eqref{eq:pythagorean_main} is
consequential. The same quantity \(D_{\mathrm{KL}}(q_t\|q_{t-1}^{\star})\)
that defines the backward ball becomes the kinetic energy of the forward step.
The free-energy--kinetic-energy relation is therefore exact; no small-step
symmetry approximation between the two orientations of KL divergence is
required.

\subsection{Intermediate Steps vs. Cost of a Direct Jump}

Consider the final destination \(q_T=p_{\mathrm{pref}}\). A direct one-step
transition from equilibrium consumes
\begin{equation}
K_{\mathrm{direct}}
=
D_{\mathrm{KL}}(p_{\mathrm{pref}}\|q_0).
\end{equation}
Introducing an optimally chosen predecessor \(q_{T-1}^{\star}\) and applying
\eqref{eq:pythagorean_main} gives
\begin{equation}
D_{\mathrm{KL}}(p_{\mathrm{pref}}\|q_0)
\geq
D_{\mathrm{KL}}(p_{\mathrm{pref}}\|q_{T-1}^{\star})
+
D_{\mathrm{KL}}(q_{T-1}^{\star}\|q_0).
\label{eq:direct_vs_intermediate}
\end{equation}
Thus, the direct KL cost dominates the two information-geometric components
created by the optimal intermediate point. This provides the basic motivation
for resolving a large displacement into a sequence of information projections as the system moves from the Gibbs distribution $q_0$ to the given final distribution $p_\text{pref}$.
Note that, this statement is not true for an arbitrary subdivision; it is the
projection geometry that produces \eqref{eq:direct_vs_intermediate}.

\subsection{Accumulated Energy Relation}

Applying Theorem~\ref{thm:pythagorean_free_energy} to every step gives
\begin{equation}
F_r(q_t\|q_{t-1}^{\star})\geq K_t,
\qquad t=1,\ldots,T.
\end{equation}
Summing along the path,
\begin{equation}
\sum_{t=1}^{T}F_r(q_t\|q_{t-1}^{\star})
\geq
\sum_{t=1}^{T}K_t.
\label{eq:accumulated_energy_inequality}
\end{equation}
The free-energy increments telescope, and for \(q_T=p_{\mathrm{pref}}\),
\begin{equation}
\sum_{t=1}^{T}F_r(q_t\|q_{t-1}^{\star})
=
D_{\mathrm{KL}}(p_{\mathrm{pref}}\|q_0).
\end{equation}
Consequently,
\begin{equation}
\sum_{t=1}^{T}K_t
\leq
D_{\mathrm{KL}}(p_{\mathrm{pref}}\|q_0).
\label{eq:energy_conservation_bound}
\end{equation}
Equation~\eqref{eq:energy_conservation_bound} gives an
energy-conservation-like interpretation of the information-geometric path. \cite{kolchinsky2021dependence}
The free-energy separation between equilibrium and the target sets the total
energetic scale of the transition, while the cumulative kinetic expenditure
along a sequence of information projections cannot exceed this quantity. At
each step, the increase in stored nonequilibrium free energy is sufficient to
account for the informational kinetic cost of the transition.
One can restore the physical units by multiplying both sides of \eqref{eq:energy_conservation_bound} with $k_B \tau$ to convert the right-hand side to the nonequilibrium free-energy difference between
the preferred state and Gibbs equilibrium.
As a result, informational kinetic expenditure along the path is converted into the free energy accumulated in constructing the nonequilibrium state.

%% file: 4-problem-statement-revised.tex
\section{Problem Statement}
\label{sec:problem}

We consider the forward evolution of a distribution on the probability simplex
from the Gibbs equilibrium \(q_0\) to a prescribed nonequilibrium distribution
\(p_{\mathrm{pref}}\).  The trajectory is represented by
\begin{equation}
    q_0 \rightarrow q_1 \rightarrow \cdots \rightarrow q_T
    = p_{\mathrm{pref}},
    \label{eq:forward_path}
\end{equation}
where each transition is subject to the kinetic-energy constraint
\begin{equation}
    K_t
    =
    D_{\mathrm{KL}}(q_t\|q_{t-1})
    \leq \delta,
    \qquad t=1,\ldots,T.
    \label{eq:kinetic_constraint_problem}
\end{equation}
Here, \(\delta>0\) specifies the maximum informational displacement that can be
realized in a single step.

The terminal time is not prescribed a priori.  Instead, we seek the smallest
number of admissible transitions required to reach the preferred distribution.
This leads to the minimum-hitting-time problem
\begin{equation}
\begin{aligned}
    T_\delta^\star
    \triangleq
    \min_{\substack{T\in\mathbb N,\\
                    q_1,\ldots,q_{T-1}\in\Delta(\mathcal X)}}
    \quad & T \\
    \mathrm{subject\ to}\quad
    & q_T=p_{\mathrm{pref}},\\
    & D_{\mathrm{KL}}(q_t\|q_{t-1})\leq\delta,
      \qquad t=1,\ldots,T .
\end{aligned}
\label{eq:global_problem}
\end{equation}
Problem~\eqref{eq:global_problem} is a minimum-hitting-time problem on the
probability simplex under an information-geometric kinetic constraint.  Unlike
a Euclidean shortest-path problem, KL divergence does not satisfy a triangle
inequality, and the optimal trajectory is therefore not characterized solely
by the endpoint divergence
\(D_{\mathrm{KL}}(p_{\mathrm{pref}}\|q_0)\).  In particular, a locally optimal
move need not, without further structure, coincide with the globally
minimum-time policy.

The Pythagorean geometry developed in Section~\ref{sec:pythagorean}
nevertheless provides a natural constructive principle for generating
admissible paths.  In the following section, we introduce a sequence of
information projections that maximizes the free-energy progress associated with
each admissible step, derive its explicit form, and establish a finite-time
bound for the resulting trajectory.  The projection rule is therefore used as
a geometrically canonical local construction with provable energy and
hitting-time guarantees, without requiring an a priori claim of global
minimum-time optimality over all admissible trajectories.

%% file: 5-approach-least-action-information-projection.tex
\section{Approach: Least-Action Path via Information Projection}
\label{sec:approach}

The minimum-hitting-time problem in Section~\ref{sec:problem} does not, in
general, reduce directly to a myopic optimization.  The Pythagorean structure
of Section~\ref{sec:pythagorean}, however, suggests a canonical local policy:
given a point \(q_t\), choose among all admissible predecessors the one that
maximizes the relative free-energy gain while preserving reachability of
\(q_t\) in a single kinetic step.  Starting from
\(q_T=p_{\mathrm{pref}}\), this construction proceeds backward toward the Gibbs
reference \(q_0\); reversing the resulting sequence gives the forward physical
trajectory.

For a given \(q_t\), define the greedy predecessor by
\begin{align}
\label{eq:greedy_fr_problem-obj}
q_{t-1}^{\star}
&= \underset{q\in\Delta(\mathcal X)}{\text{argmax}} \ ~ F_r(q_t\|q) \\
&\mathrm{subject\ to}
\ ~
D_{\mathrm{KL}}(q_t\|q)\leq\delta.
\label{eq:greedy_fr_problem-constraint}
\end{align}
Since
\[
F_r(q_t\|q)
=
D_{\mathrm{KL}}(q_t\|q_0)
-
D_{\mathrm{KL}}(q\|q_0),
\]
the first term is fixed for a given \(q_t\), and
\eqref{eq:greedy_fr_problem-obj}, \eqref{eq:greedy_fr_problem-constraint} is equivalently
\begin{align}
\label{eq:greedy_projection_problem-obj}
q_{t-1}^{\star}
&= \underset{q\in\Delta(\mathcal X)}{\text{argmin}} \ ~ D_{\mathrm{KL}}(q\|q_0) \\
&\mathrm{subject\ to}
\ ~
D_{\mathrm{KL}}(q_t\|q)\leq\delta.
\label{eq:greedy_projection_problem-constraint}
\end{align}
Thus, among all distributions from which \(q_t\) is reachable in one admissible
forward step, the rule selects the predecessor with minimum nonequilibrium free
energy relative to Gibbs equilibrium.  It is therefore locally optimal in
free-energy progress, although we do not require it to coincide with the
globally minimum-hitting-time policy of \eqref{eq:global_problem}.

\subsection{Action and the Local Least-Action Principle}

The constrained problem \eqref{eq:greedy_fr_problem-obj}, \eqref{eq:greedy_fr_problem-constraint} admits a direct
least-action interpretation. Let \(\lambda_t\geq 0\) denote the Lagrange
multiplier associated with the kinetic constraint.  The Lagrangian for the
maximization problem is
\begin{equation}
\mathcal L_t(q,\lambda_t)
=
F_r(q_t\|q)
-
\lambda_t
\left[
D_{\mathrm{KL}}(q_t\|q)-\delta
\right].
\label{eq:greedy_lagrangian_max}
\end{equation}
The additive term \(\lambda_t\delta\) does not depend on \(q\).  Hence maximizing
\eqref{eq:greedy_lagrangian_max} is equivalent to minimizing
\begin{equation}
\mathcal A_t(q;q_t)
\triangleq
\lambda_t
\underbrace{D_{\mathrm{KL}}(q_t\|q)}_{\text{kinetic energy}} ~~
-
\underbrace{F_r(q_t\|q)}_{\text{relative free-energy}}.
\label{eq:local_action}
\end{equation}
We refer to \(\mathcal A_t\) as the \textbf{instantaneous information-geometric
action}.  It has the same kinetic-minus-potential algebraic structure as the
classical Lagrangian: the first term penalizes motion, while the second rewards
the free-energy increase realized by the transition.  Since
\(D_{\mathrm{KL}}(q_t\|q_0)\) is constant with respect to \(q\),
\eqref{eq:local_action} is equivalently
\begin{equation}
\mathcal A_t(q;q_t)
=
D_{\mathrm{KL}}(q\|q_0)
+
\lambda_t D_{\mathrm{KL}}(q_t\|q)
-
D_{\mathrm{KL}}(q_t\|q_0).
\label{eq:local_action_expanded}
\end{equation}
Therefore the greedy least-action step is obtained from
\begin{equation}
q_{t-1}^{\star}
= \underset{q\in\Delta(\mathcal X)}{\text{argmin}} ~ 
\left\{
D_{\mathrm{KL}}(q\|q_0)
+
\lambda_t D_{\mathrm{KL}}(q_t\|q)
\right\}.
\label{eq:least_action_unconstrained}
\end{equation}
The action in \eqref{eq:local_action} is dimensionless under our normalization.
Restoring energy units multiplies both energetic terms by \(k_B\tau\); if one also
associates a physical duration \(\Delta\) with each discrete update, then
\(k_B\tau \,\Delta\,\mathcal A_t\) has the dimensions of mechanical action.
Our use of ``least action'' refers to the variational structure of each
information-geometric step rather than to an assertion of microscopic
Hamiltonian dynamics.

\subsection{Solving for the Least Action Path}
\label{subsec:lambert_solution}

The orientation of the KL divergences in
\eqref{eq:least_action_unconstrained} is dictated by the thermodynamic and
Pythagorean construction.  Unlike the exponential-tilting solution obtained
when the current distribution occupies the first argument of both divergences,
the present problem leads naturally to the Lambert \(W\) function.

\begin{theorem}[Optimal information projection]
\label{thm:lambert_projection}
Assume \(q_0(x)>0\) and \(q_t(x)>0\) for all \(x\in\mathcal X\).  If
\(
D_{\mathrm{KL}}(q_t\|q_0)\leq\delta,
\)
then \(q_{t-1}^{\star}=q_0\).  Otherwise, the kinetic constraint is active and
the unique solution of \eqref{eq:greedy_projection_problem-obj}, \eqref{eq:greedy_projection_problem-constraint} is
\begin{equation}
q_{t-1}^{\star}(x) =
\frac{q_t(x)}
{W_0\!\left(
\alpha_t\,\frac{q_t(x)}{q_0(x)}
\right)} \cdot \frac{1}{Z_0},
\label{eq:lambert_solution}
\end{equation}
where
\[ Z_0= \sum_{y\in\mathcal X} \frac{q_t(y)}
{W_0\!\left(\alpha_t\,\frac{q_t(y)}{q_0(y)} \right)} \]
is the normalization constant and \(W_0(\cdot)\) is the principal branch of the Lambert \(W\) function \cite{corless1996lambert} and
\(\alpha_t>0\) is chosen so that
\begin{equation}
D_{\mathrm{KL}}(q_t\|q_{t-1}^{\star})=\delta.
\label{eq:lambert_constraint}
\end{equation}
\end{theorem}

\begin{proof} 
See Appendix.~\ref{sec:proof_thm_lambert}
\end{proof}

The solution in \eqref{eq:lambert_solution} is a nonlinear deformation of
\(q_t\) toward the Gibbs reference.  The scalar \(\alpha_t\) determines how far
the projection moves while the KL constraint fixes the kinetic expenditure.
Repeated application produces a sequence of Lambert-\(W\) projections rather
than the exponential-family recursion of the earlier KL orientation.

\subsection{Finite-Time Performance Guarantee}
\label{subsec:performance_guarantee}

The Pythagorean inequality gives a simple bound on the number of full
kinetic-budget projections required by the construction.  To distinguish this
quantity from the globally optimal hitting time \(T_\delta^\star\), let $T_\delta^{\mathrm{IP}}$ be the number of steps for the information-projection (IP) based greedy solution takes.

\begin{theorem}[Finite-time bound]
\label{thm:finite_time_bound}
The number of forward transitions in the information-projection path satisfies
\begin{equation}
T_\delta^\star \leq T_\delta^{\mathrm{IP}} \leq
\left\lceil \frac{D_{\mathrm{KL}}(p_{\mathrm{pref}}\|q_0)}{\delta} \right\rceil.
\label{eq:hitting_time_bound}
\end{equation}
\end{theorem}

\begin{proof}
Let \(N_\delta\) denote the number of backward information-projection steps for which the kinetic constraint is active, i.e.,
\[ D_{\mathrm{KL}}(q_t\|q_{t-1}^{\star})=\delta, \]
excluding the final step ($q_1\to q_0$). If $0<D_{\mathrm{KL}}(p_{\mathrm{pref}}\|q_0) \leq \delta$, $N_\delta = 0$ and \eqref{eq:hitting_time_bound} is trivial. Otherwise, for each full backward projection, by
Theorem~\ref{thm:pythagorean_free_energy},
\begin{equation}
\delta = D_{\mathrm{KL}}(q_t\|q_{t-1}^{\star}) \leq
D_{\mathrm{KL}}(q_t\|q_0) - D_{\mathrm{KL}}(q_{t-1}^{\star}\|q_0).
\label{eq:delta_progress}
\end{equation}
Summing \eqref{eq:delta_progress} over the \(N_\delta\) full projections telescopes and gives
\begin{equation}
N_\delta\delta \leq D_{\mathrm{KL}}(p_{\mathrm{pref}}\|q_0) - D_{\mathrm{KL}}(q_1^{\star}\|q_0),
\label{eq:telescope_N}
\end{equation}
where \(q_1^{\star}\) is the final projected point before the Gibbs distribution
becomes directly reachable. By the stopping rule,
\[ 0< D_{\mathrm{KL}}(q_1^{\star}\|q_0) \leq\delta, \]
and therefore
\begin{equation}
N_\delta\delta < D_{\mathrm{KL}} (p_{\mathrm{pref}}\|q_0). 
\label{eq:N_delta_bound}
\end{equation}
The final transition \(q_0\rightarrow q_1^{\star}\) has a positive kinetic cost at most \(\delta\), so the complete information-projection path contains \(T_\delta^{\mathrm{IP}}=N_\delta+1\) transitions.  Since \(N_\delta\) is integer, \eqref{eq:N_delta_bound} implies
\[ T_\delta^{\mathrm{IP}} \leq \left\lceil \frac{D_{\mathrm{KL}}(p_{\mathrm{pref}}\|q_0)}{\delta} \right\rceil. \]
Finally, \(T_\delta^\star\) is the minimum hitting time over all admissible paths, while the information-projection construction provides one admissible path.  Hence \(T_\delta^\star\leq T_\delta^{\mathrm{IP}}\), completing the proof.
\end{proof}

Note that, in the bound, the numerator is the dimensionless free-energy separation between the preferred nonequilibrium state and Gibbs equilibrium, while \(\delta\) is the kinetic-energy budget available per full projection.  Restoring physical units multiplies both by \(k_B \tau\), so their ratio, and hence the step-count bound, is unchanged.

\subsection{Discussion}

The construction above separates three distinct statements.  First,
\eqref{eq:global_problem} defines the globally minimum hitting-time problem.
Second, \eqref{eq:greedy_fr_problem-obj}, \eqref{eq:greedy_fr_problem-constraint} defines a canonical local policy obtained
from information geometry: each backward step maximizes the attainable
free-energy reduction under the kinetic budget.  Third, the Pythagorean theorem
provides a performance guarantee for the resulting path, leading to the
finite-time bound in Theorem~\ref{thm:finite_time_bound}.  We therefore do not
identify the greedy information-projection path with the globally
minimum-hitting-time trajectory; rather, it is a constructive least-action path
whose energetic progress and hitting time can be controlled explicitly.

The same formulation also exposes the physical structure of the path.  Each
local optimization minimizes a kinetic-minus-relative-free-energy action, while
the Pythagorean inequality guarantees that the free-energy gain of the step is
at least sufficient to account for its kinetic expenditure.  Summed along the
trajectory, these local relations recover the energy-conservation-like bound of
Section~\ref{sec:pythagorean}.  Thus, information projection simultaneously
provides a variational rule for choosing the path, an explicit Lambert-\(W\)
update for each step, and a direct connection between thermodynamic distance
from equilibrium and the number of admissible distributional transitions
required to construct the target state.

%% file: 6-numeric-example.tex
\begin{figure}[ht]
\centering
\includegraphics[width=0.48\columnwidth]{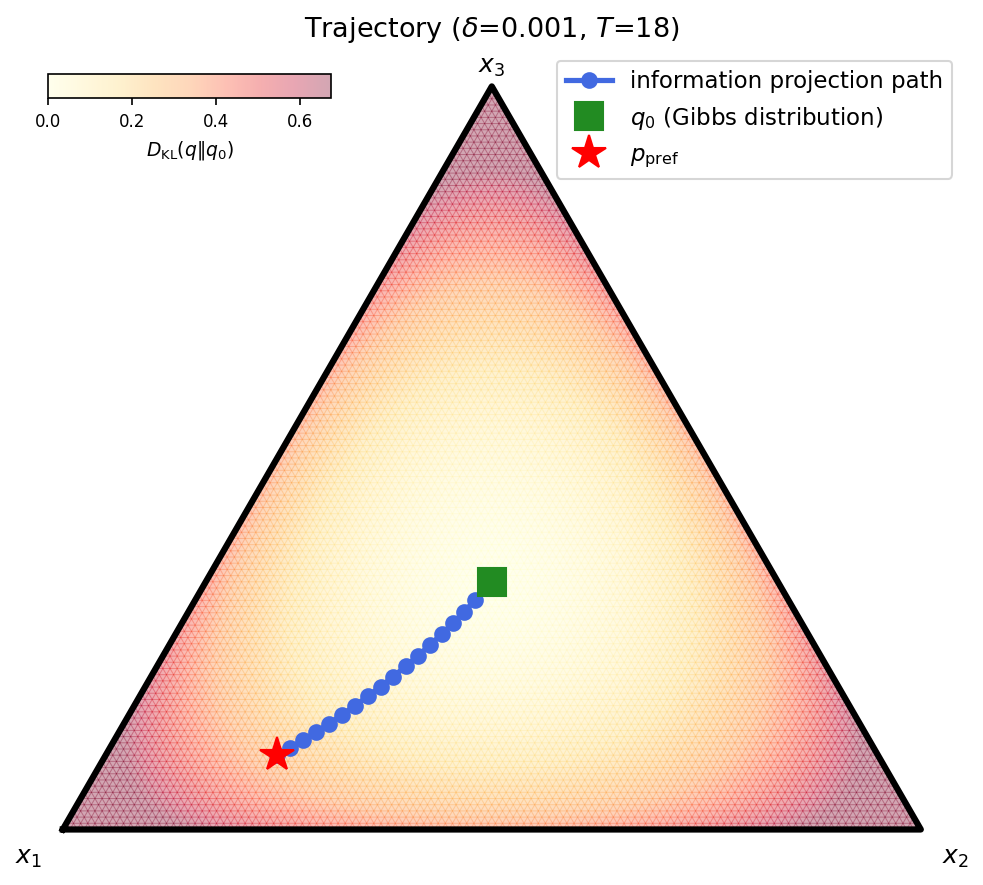}
\includegraphics[width=0.48\columnwidth]{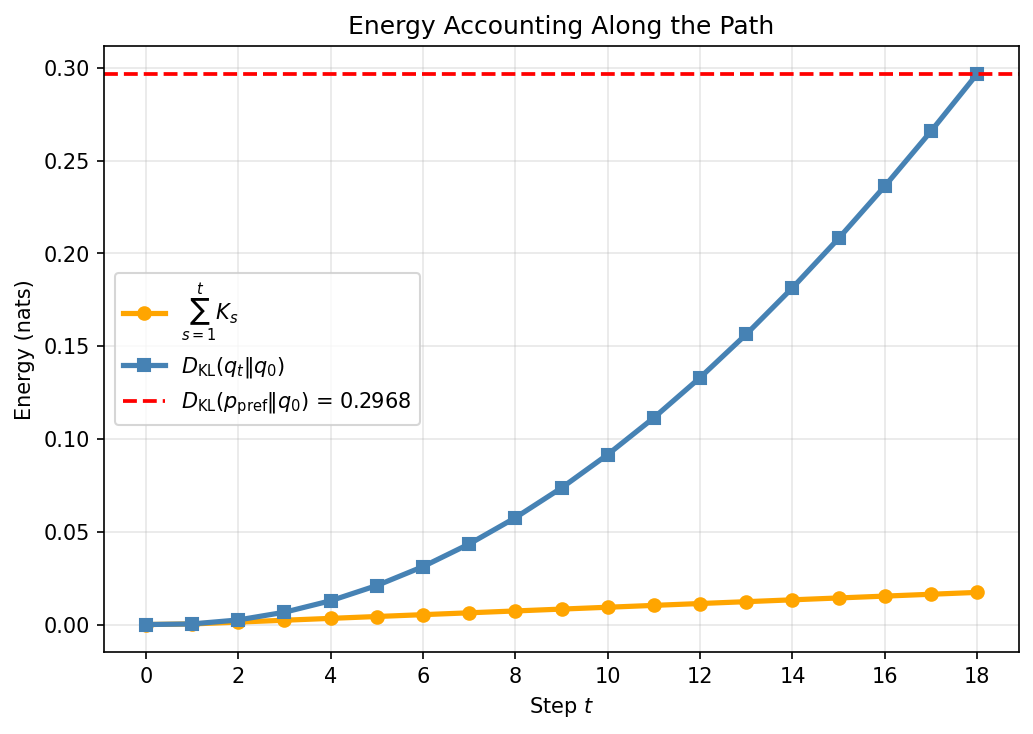}\\[4pt]
\includegraphics[width=0.48\columnwidth]{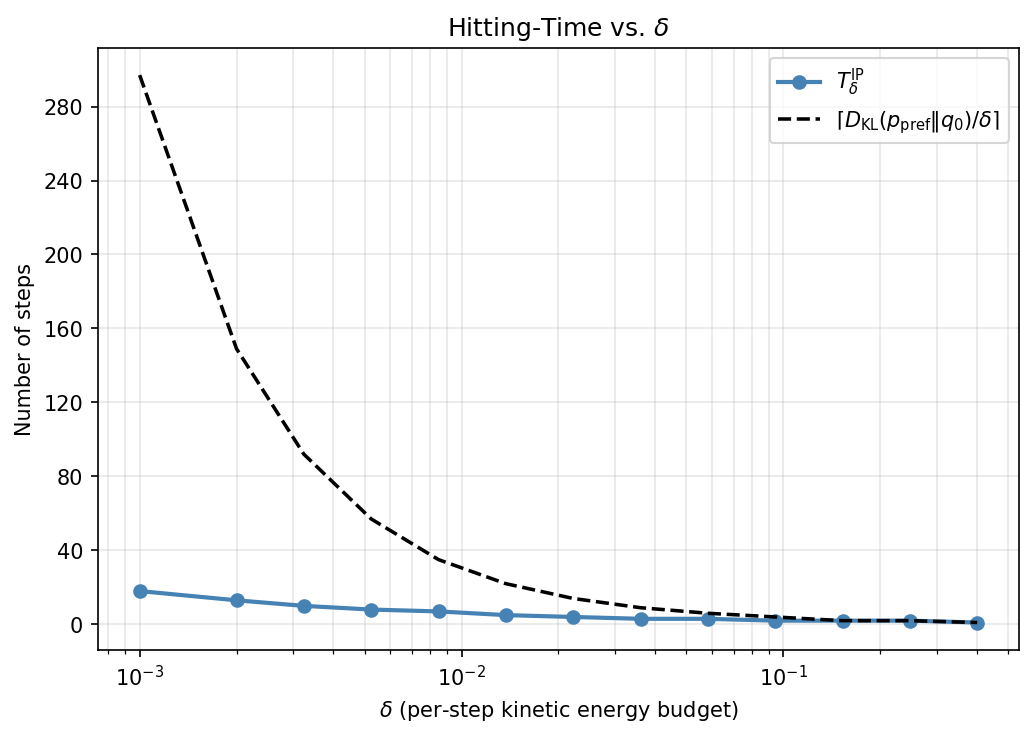}
\includegraphics[width=0.48\columnwidth]{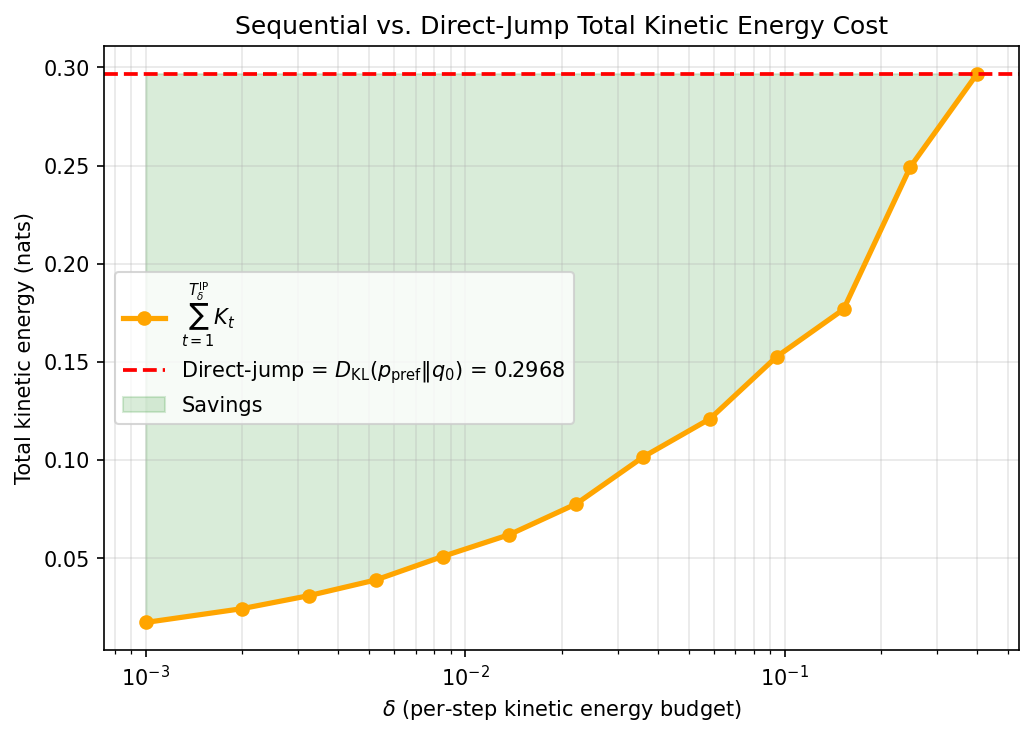}
\vspace{-0.1in}
\caption{Numerical results for $q_0 = (1/3, 1/3, 1/3)$, $p_{\mathrm{pref}} = (0.7, 0.2, 0.1)$, with $D_{\mathrm{KL}}(p_{\mathrm{pref}}\|q_0) \approx 0.297$ nats. \textbf{Top-left:} Least-action trajectory on the simplex with $\delta = 0.001$ ($T = 18$ steps) over the free-energy landscape $D_{\mathrm{KL}}(q\|q_0)$. \textbf{Top-right:} Running nonequilibrium free energy $D_{\mathrm{KL}}(q_t\|q_0)$ and cumulative kinetic cost $\sum_{s\leq t} K_s$ along the path, illustrating the energy-conservation bound $\sum_t K_t \leq D_{\mathrm{KL}}(p_{\mathrm{pref}}\|q_0)$. \textbf{Bottom-left:} Hitting time $T^{\mathrm{IP}}_\delta$ versus the upper bound $\lceil D_{\mathrm{KL}}(p_{\mathrm{pref}}\|q_0)/\delta \rceil$; the IP path is well below the bound due to $F_r - K > 0$ at each projection. \textbf{Bottom-right:} Total sequential kinetic cost $\sum_t K_t$ versus the direct-jump cost $D_{\mathrm{KL}}(p_{\mathrm{pref}}\|q_0)$, showing an order-of-magnitude reduction for small $\delta$.}
\vspace{-0.2in}
\label{fig:numeric}
\end{figure}

\section{Numerical Illustration}
\label{sec:numeric}

We illustrate the stopping-time formulation with a concrete example on the $2$-simplex $\Delta(\mathcal{X})$ with $|\mathcal{X}|=3$. Let the initial equilibrium distribution be uniform, $\pi_0 = (1/3,\, 1/3,\, 1/3)$, and the preferred distribution be $p_{\mathrm{pref}} = (0.7,\, 0.2,\, 0.1)$. The initial KL divergences are $D_{\mathrm{KL}}(\pi_0 \| p_{\mathrm{pref}}) = 0.324$ nats and $D_{\mathrm{KL}}(p_{\mathrm{pref}} \| \pi_0) = 0.297$ nats. We sweep the kinetic constraint over $\delta \in [0.001, 0.4]$ and apply the backward information-projection path defined by Eqs.~\eqref{eq:greedy_fr_problem-obj}–\eqref{eq:greedy_fr_problem-constraint}, whose closed-form solution is given by Theorem~\ref{thm:lambert_projection}. In all cases the trajectory terminates exactly at $p_{\mathrm{pref}}$.

\noindent \textbf{Trajectory and energy evolution:}
Figure~\ref{fig:numeric} (top-left) shows the least-action path on the simplex for $\delta = 0.001$. The trajectory traces a smooth curve from the centroid toward the target, consistent with the exponential tilting update of Theorem~2. Figure~\ref{fig:numeric} (top-right) displays the corresponding energy evolution: both $F_{\mathrm{rel}}$ and $P_{\mathrm{rel}}$ increase monotonically toward their maximum values, confirming steady progress toward the preferred distribution at each step.

\noindent \textbf{Effect of step size:}
Figure~\ref{fig:numeric} (bottom-left) reveals the tradeoff between step size and stopping time. As $\delta$ decreases from $0.4$ to $0.001$, the stopping time $N_\delta$ increases from $1$ to $18$, while the total kinetic cost decreases from $0.297$ to $0.017$ nats. This demonstrates that smaller steps are fundamentally more efficient, converting kinetic energy into free-energy reduction with less waste.

\noindent \textbf{Kinetic energy savings:}
Figure~\ref{fig:numeric} (bottom-right) compares the sequential kinetic cost against the direct-jump baseline $D_{\mathrm{KL}}(p_{\mathrm{pref}} \| \pi_0) = 0.297$. For all tested step sizes, the sequential strategy achieves strictly lower total kinetic energy. At $\delta = 0.001$, the savings exceed $94\%$, demonstrating the dominance of incremental motion established by the Pythagorean theorem (Theorem~\ref{thm:pythagorean_free_energy}).

%% file: 7-geodesic-cost-variational-distribution.tex
\section{Geodesic Cost on Variational Distribution}
\label{sec:geodesic_cost}

The development so far assumes that motion on the probability simplex is governed
only by the tradeoff between free-energy progress and the kinetic cost of changing
distributions. In many applications, however, there may be an additional cost
associated with occupying particular states or traversing particular regions of
the simplex. Such costs may represent additional energy penalties, risk, or application-specific preferences
over the intermediate distributions visited by the system.

In this section, we show that such costs can be incorporated without changing
the essential structure of the least-action construction. An expected state
cost induces a scalar cost landscape over the simplex and augments the
Gibbs-referenced free energy associated with each variational distribution.
Remarkably, this additional term can be absorbed exactly into the reference
measure through an exponential tilt of the Gibbs distribution. The resulting
problem therefore retains the same information-geometric structure as the
original least-action formulation, while admitting a direct physical
interpretation in which the added cost acts as an external contribution to the
underlying energy landscape.

\subsection{Greedy Least-Action Problem with Geodesic Cost}
\label{subsec:geodesic_greedy}

Let \(c(x)\geq 0\) denote a state-dependent cost associated with occupying
\(x\in\mathcal X\). For a variational distribution \(q\), define
\begin{equation}
    C(q)
    \triangleq
    \mathbb E_q[c(X)]
    =
    \sum_{x\in\mathcal X}q(x)c(x).
    \label{eq:expected_geodesic_cost}
\end{equation}
Although \eqref{eq:expected_geodesic_cost} is written as an expected state cost, it induces a cost function over \textit{distributions} \(q\) on the probability simplex and can therefore be used to penalize trajectories that pass through undesirable regions of the simplex.

It is useful to incorporate this term directly into the free-energy landscape.
Define the cost-augmented dimensionless free energy
\begin{equation}
    \Phi_\gamma(q)
    \triangleq
    D_{\mathrm{KL}}(q\|q_0)
    +
    \gamma C(q),
    \label{eq:augmented_free_energy}
\end{equation}
where \(\gamma\geq 0\) controls the importance of the additional cost. The
corresponding relative free-energy gain from \(q\) to \(q_t\) is
\begin{align}
    F_{r,\gamma}(q_t\|q)
    &\triangleq
    \Phi_\gamma(q_t)-\Phi_\gamma(q)
    \nonumber\\
    &=
    F_r(q_t\|q)
    +
    \gamma\left[C(q_t)-C(q)\right].
    \label{eq:augmented_relative_free_energy}
\end{align}
For a fixed \(q_t\), the greedy predecessor therefore solves
\begin{equation}
    q_{t-1}^{\star}
    =
    \underset{q\in\Delta(\mathcal X)}{\text{argmax}}  ~
    F_{r,\gamma}(q_t\|q)
    \quad
    \mathrm{subject\ to}
    \quad
    D_{\mathrm{KL}}(q_t\|q)\leq\delta.
    \label{eq:geodesic_greedy_problem}
\end{equation}
Since the terms \(\Phi_\gamma(q_t)\) and \(C(q_t)\) are fixed with respect to
\(q\), \eqref{eq:geodesic_greedy_problem} is equivalently
\begin{equation}
    q_{t-1}^{\star}
    =
    \underset{q\in\Delta(\mathcal X)}{\text{argmin}} ~
    \left\{
        D_{\mathrm{KL}}(q\|q_0)
        +
        \gamma C(q)
    \right\}
    \quad
    \mathrm{subject\ to}
    \quad
    D_{\mathrm{KL}}(q_t\|q)\leq\delta.
    \label{eq:geodesic_projection_problem}
\end{equation}

Introducing the multiplier \(\lambda_t\geq0\) for the kinetic constraint gives,
up to terms independent of \(q\),
\begin{align}
    \mathcal A_{t,\gamma}(q;q_t)
    &=
    \lambda_t D_{\mathrm{KL}}(q_t\|q)
    -
    F_{r,\gamma}(q_t\|q)
    \label{eq:geodesic_action}\\
    &\equiv
    D_{\mathrm{KL}}(q\|q_0)
    +
    \gamma C(q)
    +
    \lambda_t D_{\mathrm{KL}}(q_t\|q).
    \label{eq:geodesic_action_expanded}
\end{align}
Thus, the geodesic/state cost simply augments the free-energy part of the local
action. The least-action structure itself is unchanged: the local variational
problem continues to balance a KL kinetic term against progress in the relevant
free-energy landscape.

\subsection{Reduction to the Original Problem}
\label{subsec:geodesic_reduction}

The first two terms in \eqref{eq:geodesic_action_expanded} can be combined
exactly:
\begin{align}
    D_{\mathrm{KL}}(q\|q_0)+\gamma C(q)
    &=
    \sum_x q(x)
    \log\frac{q(x)}{q_0(x)e^{-\gamma c(x)}}.
    \label{eq:cost_combine}
\end{align}
Define the cost-adjusted reference distribution
\begin{equation}
    \widetilde q_0(x)
    \triangleq
    \frac{
        q_0(x)e^{-\gamma c(x)}
    }{
        \displaystyle
        \sum_{y\in\mathcal X}q_0(y)e^{-\gamma c(y)}
    }.
    \label{eq:tilted_gibbs}
\end{equation}
If
\begin{equation}
    Z_\gamma
    \triangleq
    \sum_{y\in\mathcal X}q_0(y)e^{-\gamma c(y)},
\end{equation}
then
\begin{equation}
    D_{\mathrm{KL}}(q\|q_0)+\gamma C(q)
    =
    D_{\mathrm{KL}}(q\|\widetilde q_0)
    -
    \log Z_\gamma.
    \label{eq:geodesic_kl_equivalence}
\end{equation}
The final term is independent of \(q\). Hence
\eqref{eq:geodesic_projection_problem} reduces exactly to
\begin{equation}
    q_{t-1}^{\star}
    =
\underset{q\in\Delta(\mathcal X)}{\text{argmin}} ~
    D_{\mathrm{KL}}(q\|\widetilde q_0)
    \quad
    \mathrm{subject\ to}
    \quad
    D_{\mathrm{KL}}(q_t\|q)\leq\delta.
    \label{eq:geodesic_reduced_problem}
\end{equation}
Therefore, the geodesic-cost extension has precisely the same mathematical form
as the original information-projection problem, with \(q_0\) replaced by the
exponentially tilted reference \(\widetilde q_0\).

In particular, the Lambert-\(W\) solution of
Theorem~\ref{thm:lambert_projection} carries over immediately:
\begin{equation}
q_{t-1}^{\star}(x) =
\frac{q_t(x)} {W_0\!\left( \alpha_t\,\frac{q_t(x)}{\widetilde q_0(x)} \right)} \cdot \frac{1}{\tilde{Z}_0} ,
\label{eq:geodesic_lambert_solution}
\end{equation}
where 
\[ \tilde{Z}_0 = \sum_{y\in\mathcal X}
\frac{q_t(y)} {W_0\!\left( \alpha_t\,\frac{q_t(y)}{\widetilde q_0(y)} \right)} \]
and \(\alpha_t>0\) is selected so that the kinetic constraint is active, unless the cost-adjusted reference itself is already feasible.

A state-dependent cost therefore reshapes the Gibbs reference against which nonequilibrium free energy is measured. Algebraically, this follows from the fact that the variational distribution \(q\) appears in the first argument of \(D_{\mathrm{KL}}(q|q_0)\), allowing the additive expected cost to be absorbed into \(q_0\) through an exponential tilt. The resulting reference distribution thus represents the effective equilibrium associated with the augmented energy
landscape.

\subsection{Implications and Physical Interpretation}
\label{subsec:geodesic_implications}

Equation~\eqref{eq:tilted_gibbs} has a particularly direct thermodynamic
interpretation. Suppose
\begin{equation}
    q_0(x)
    =
    \frac{e^{-\beta E(x)}}{Z_0},
    \qquad
    \beta=\frac{1}{k_B\tau},
    \label{eq:gibbs_original}
\end{equation}
is the Gibbs distribution associated with energy \(E(x)\). Then
\begin{equation}
    \widetilde q_0(x)
    \propto
    \exp\!\left[-\beta E(x)-\gamma c(x)\right].
    \label{eq:effective_gibbs}
\end{equation}
Thus, the added cost acts exactly as an additional contribution to the energy
landscape. If \(c(x)\) is itself expressed in physical energy units,
the dimensionless term \(\gamma c(x)\) in the equations above is replaced by
\(\beta\gamma c(x)\), giving
\begin{equation}
    \widetilde q_0(x)
    \propto
    \exp\!\left\{
        -\beta\left[E(x)+\gamma c(x)\right]
    \right\}.
    \label{eq:effective_hamiltonian}
\end{equation}
In this case, the geodesic penalty can be interpreted as an external potential
that modifies the effective Hamiltonian from \(E(x)\) to
\(E(x)+\gamma c(x)\). \cite{todorov2006linearly}

This gives a sharper physical interpretation than a generic path penalty.
Expensive states are exponentially suppressed in the effective equilibrium
distribution, and the least-action trajectory is correspondingly deformed away
from regions that place large probability mass on those states. The geometry of
the local optimization does not need to be redesigned: the external cost is
absorbed into the thermodynamic reference distribution, after which the same
Pythagorean projection and Lambert-\(W\) machinery applies.

There is also an important endpoint interpretation. If the additional cost is
regarded as a genuine physical contribution to the system energy that remains
present throughout the process, then \(\widetilde q_0\) is the appropriate Gibbs
equilibrium of the modified system, and the entire framework carries over with
\(q_0\) replaced by \(\widetilde q_0\). If, instead, \(q_0\) is required to
remain a fixed physical endpoint and \(c(x)\) represents only a transient
routing or safety penalty, then \eqref{eq:geodesic_kl_equivalence} should be
understood as a local algebraic equivalence: the intermediate projections are
biased by the tilted reference, while the original endpoint constraint at
\(q_0\) must still be imposed explicitly.

More generally, the result shows that state costs and thermodynamic reference
measures are not independent objects in this variational geometry. For linear
expected costs of the form \eqref{eq:expected_geodesic_cost}, introducing a cost
landscape is equivalent to changing the reference measure by exponential
tilting. In the present orientation, the cost therefore reshapes the
\emph{equilibrium landscape} rather than the preferred destination. This is the
natural counterpart of the thermodynamic interpretation developed in
Section~\ref{sec:model}: external energetic penalties modify the Gibbs measure,
and the least-action path responds through the same information-geometric
projection structure.

%% file: 8-discussion-conclusion.tex
\section{Conclusion and Future Work}
\label{sec:discussion}

This work develops a discrete least-action description of motion on the
probability simplex by combining three structures that are usually treated
separately: thermodynamic free energy, information-geometric kinetic cost, and
information projection. With the Gibbs distribution \(q_0\) as the equilibrium
reference, the quantity \(D_{\mathrm{KL}}(p\|q_0)\) is the dimensionless
nonequilibrium free-energy excess of a distribution \(p\), while
\(D_{\mathrm{KL}}(q_t\|q_{t-1})\) measures the informational cost of a
transition. The orientation of these divergences is central to the
construction because it allows the Pythagorean geometry of relative entropy to
connect the free-energy gain of a step directly to the kinetic energy expended
in that step.

The resulting picture has a simple physical interpretation. The Gibbs state
is the zero-free-energy reference, whereas a prescribed distribution
\(p_{\mathrm{pref}}\) represents a nonequilibrium state carrying excess free
energy \(k_B \tau D_{\mathrm{KL}}(p_{\mathrm{pref}}\|q_0)\). A path from
equilibrium to this state can be regarded as a sequence of elementary
distributional transitions through which this nonequilibrium structure is
created. For the information-projection path, the Pythagorean inequality
guarantees at every step that the increase in relative free energy is at least
as large as the corresponding KL kinetic expenditure. Summed along the path,
this gives the energy-conservation-like relation
\begin{equation}
    \sum_t K_t \leq D_{\mathrm{KL}}(p_{\mathrm{pref}}\|q_0),
\end{equation}
which is an exact accounting relation induced by the
information geometry of the simplex rather than a microscopic mechanical conservation law. The aim is to provide a sense in which the kinetic expenditure of the path is supported by the free energy accumulated in the final nonequilibrium state.

The same construction provides a variational notion of action. At each step,
the constrained free-energy maximization can be written as minimization of a
kinetic-minus-free-energy Lagrangian. This gives the local optimization the
same algebraic structure as the Lagrangian of classical mechanics, while the
probability simplex and KL divergence replace the configuration space and
quadratic kinetic energy of the mechanical setting. 
The corresponding optimality condition leads to a Lambert-\(W\) update, which characterizes each information projection explicitly.

It is important to distinguish this local least-action construction from the
globally minimum-hitting-time problem. The latter asks for the smallest number
of KL-constrained transitions connecting \(q_0\) to
\(p_{\mathrm{pref}}\), whereas the information-projection policy makes
maximal local free-energy progress subject to the kinetic budget. We have not
assumed that these two problems are globally equivalent on an unrestricted
simplex. Instead, the projection construction provides a canonical feasible
path with explicit energetic and finite-time guarantees. In particular, the
number of full kinetic-budget projections is bounded by the ratio between the
endpoint free-energy separation and the per-step kinetic budget. Establishing
conditions under which the local rule is also globally minimum-time optimal
would sharpen the connection to dynamic programming and remains an important
direction for future work.

The extension to state-dependent geodesic costs reveals an additional physical
consequence of the framework. An expected cost assigned to intermediate
distributions does not require a different projection geometry; it can be
absorbed into a deformation of the Gibbs reference. Physically, the added cost
acts as an external contribution to the energy landscape, suppressing
high-cost states in the effective equilibrium distribution and thereby
reshaping the least-action path. This gives application-specific constraints and/or penalties
a direct thermodynamic interpretation. More broadly, it shows that the reference measure and the cost landscape are not independent objects: for linear expected costs, changing one is equivalent to exponentially tilting the other.

This observation also suggests a useful distinction between two kinds of
extensions. If the added cost represents a genuine persistent contribution to
the physical energy, then the tilted Gibbs distribution is the appropriate
equilibrium reference for the modified system. If instead the cost is only a
transient routing, risk, or control penalty, the tilt should be interpreted as
a local variational device while the original Gibbs state remains the physical
endpoint. Understanding how these two interpretations interact with more
general nonlinear or path-dependent costs is another natural direction for
extending the framework.

Several further questions follow from the present formulation. A
continuous-time limit could connect the discrete KL action developed here to
Fisher--Rao geometry, thermodynamic length, and continuous variational
principles \cite{ito2023geometric}. The backward information-projection recursion may also admit
additional structure under restricted statistical families, where global
minimum-time optimality or simpler composition laws could emerge. Beyond the
linear state costs considered here, nonlinear functionals of the variational
distribution may lead to richer effective energy landscapes while preserving
parts of the same geometric construction. Finally, numerical studies can
clarify how the Lambert-\(W\) path compares with alternative admissible paths in
kinetic expenditure, step count, and sensitivity to the local budget
\(\delta\).

Overall, the framework identifies a direct bridge between thermodynamics and
information geometry: distance from Gibbs equilibrium is free energy, motion
between distributions carries an informational kinetic cost, and information
projection supplies the variational mechanism relating the two. The resulting
least-action path provides not only a geometric rule for moving between
distributions but also a physical account of the resources required to create
and navigate nonequilibrium statistical structure.

%% file: 9-proof_solution.tex
\section{Proof of Theorem~\ref{thm:lambert_projection}}
\label{sec:proof_thm_lambert}

When \(D_{\mathrm{KL}}(q_t\|q_0)\leq\delta\), the Gibbs distribution itself is
feasible in \eqref{eq:greedy_projection_problem-constraint}; since
\(D_{\mathrm{KL}}(q\|q_0)\geq0\), with equality only at \(q=q_0\), it is the
unique optimizer.

Now suppose \(D_{\mathrm{KL}}(q_t\|q_0)>\delta\).  The constraint must then be
active.  Dropping terms independent of \(q\), the least-action problem is
\[
\min_{q\in\Delta(\mathcal X)}
\left\{
D_{\mathrm{KL}}(q\|q_0)
+
\lambda_t D_{\mathrm{KL}}(q_t\|q)
\right\}.
\]
Introduce a multiplier \(\nu_t\) for \(\sum_x q(x)=1\).  The stationarity
condition is
\begin{equation}
\log\frac{q(x)}{q_0(x)}
+1+\nu_t
-
\lambda_t\frac{q_t(x)}{q(x)}
=0.
\label{eq:lambert_stationarity}
\end{equation}
Let
\[
z_x
=
\lambda_t\frac{q_t(x)}{q(x)}.
\]
Then \(q(x)=\lambda_t q_t(x)/z_x\), and
\eqref{eq:lambert_stationarity} yields
\begin{equation}
z_x e^{z_x}
=
\alpha_t\frac{q_t(x)}{q_0(x)},
\qquad
\alpha_t
\triangleq
\lambda_t e^{1+\nu_t}.
\end{equation}
Hence
\[
z_x
=
W_0\!\left(
\alpha_t\frac{q_t(x)}{q_0(x)}
\right),
\]
and therefore
\begin{equation}
q(x)
=
\frac{
\lambda_t q_t(x)
}{
W_0\!\left(
\alpha_t\frac{q_t(x)}{q_0(x)}
\right)
}.
\end{equation}
Normalization eliminates \(\lambda_t\) and gives
\eqref{eq:lambert_solution}.  The scalar \(\alpha_t\) is determined by the
active constraint \eqref{eq:lambert_constraint}.  Strict convexity of
\(D_{\mathrm{KL}}(q\|q_0)\), together with convexity of
\(D_{\mathrm{KL}}(q_t\|q)\) in \(q\), guarantees uniqueness.